\documentclass[11pt]{elsarticle}
\usepackage{amsmath,amsthm,latexsym,amssymb,amsfonts,parskip,epsfig}

\newtheorem{lemm}{Lemma}[section]
\newtheorem{prop}[lemm]{Proposition}

\newtheorem{defi}[lemm]{Definition}

\newcommand{\R}{\mathbb{R}}                  

\newcommand{\expec}[1]{\langle #1 \rangle}
\newcommand{\dmal}{\text{d}\mu_{\text{AL}}}
\newcommand{\mal}{\mu_{\text{AL}}}

\newcommand{\wh}[1]{\widehat{#1}}

\DeclareMathOperator{\abar}{\overline{\mathcal{A}}}

\DeclareMathOperator{\knot}{K}
\DeclareMathOperator{\cyl}{Cyl}
\DeclareMathOperator{\lnk}{Link}

\begin{document}
\title{Abelian Chern-Simons theory, Stokes' theorem, and generalized connections\tnoteref{pre}}
\tnotetext[pre]{Preprint KA-TP-09-2010}

\author[apctp_postech]{Hanno Sahlmann\fnref{workplace}}
\ead{sahlmann@apctp.org}
\fntext[workplace]{The work on this article was carried out while the author was member of the Institute for Theoretical Physics, Karlsruhe University, Karlsruhe (Germany).}

\author[erlangen]{Thomas Thiemann}
\ead{thiemann@theorie3.physik.uni-erlangen.de}
\address[apctp_postech]{Asia Pacific Center for Theoretical Physics, Pohang (Korea)\\
Department of Physics, Pohang University of Science and Technology, Pohang (Korea)}
\address[erlangen]{Institute for Theoretical Physics III, Erlangen University (Germany)\\
Max Planck Institute for Gravitational Physics, Potsdam (Germany)\\
Perimeter Institute for Theoretical Physics, Waterloo (Canada)}

\begin{abstract}
Generalized connections and their calculus have been developed
in the context of quantum gravity. Here we apply them to abelian Chern-Simons
theory. We derive the expectation values of holonomies in U(1) Chern-Simons
theory using Stokes' theorem, flux operators and generalized connections. 
A framing of the holonomy loops arises in our construction, and we
show how, by choosing natural framings, the resulting expectation values
nevertheless define a functional over gauge invariant cylindrical functions. 

The abelian theory considered in the present article is the test case for our
method. It can also be applied to the non-abelian theory. Results will be reported in a companion article.
\end{abstract}

\maketitle
\section{Introduction}
One of the pillars of loop quantum gravity is a sophisticated theory of 
functions on spaces of connections with compact gauge groups. It comprises
measure theory \cite{Ashtekar:1994mh}, the definition of functional derivatives \cite{Ashtekar:1994wa}, and  \emph{spin
networks}, generalizations of Wilson loop functionals. Of particular
importance to loop quantum gravity is the definition of a (Lebesgue-like)
gauge invariant measure on such spaces. It is a natural question whether
this formalism can also be applied to other gauge theories. This has been
answered affirmatively in a number of cases, such as 2d Yang-Mills theory, 
or Maxwell theory in four dimensions. With the present article we want to add
U(1) Chern-Simons theory on $\R^3$ to this list. In particular we are interested
in two (related) questions: (i) can one make sense of the path
integral for this theory as a mathematical object related to generalized
connections? and (ii) can the formalism be used to \textit{derive} expectation
values for the path integral? The point of the present article is to
affirmatively answer both questions. 

We should point out that U(1) quantum Chern-Simons theory (see for example \cite{Polychronakos:1989cd}) in and of itself is not 
too interesting. We study the abelian theory as a test case for the underlying
methods. The application we have in mind is Chern-Simons theory for a
non-abelian gauge group, which is more relevant, both from a mathematical and from a physical perspective.  In fact, there is already very interesting work on the relation between non-abelian Chern-Simons theory and loop quantum gravity: see for example \cite{Gambini:1997fn,Noui:2004iy,Freidel:2004nb}. We report our own results in this direction in the companion article \cite{Sahlmann:2011uh}. 

U(1) Chern-Simons theory on $\R^3$ is defined by the action 
\begin{equation*}
S_\text{CS}[A]=\frac{k}{4\pi} \int_{\R^3}
A\wedge \text{d}A
\end{equation*}
where $A$ is a $U(1)$ connection. The expectation values for a collection of
holonomies $h_{\alpha_i}$ along non-intersecting loops
$\alpha_1,\ldots,\alpha_n$ can be calculated in various ways to be 
\begin{equation}
\label{eq_first}
\expec{\prod_i (h_{\alpha_i})^{n_i}}=\exp\left[-\frac{\pi i}{k}\left(\sum_j
n_{j}^2\lnk(\alpha_j,\alpha_j)+2\sum_{j<l}n_jn_l\lnk(\alpha_j,\alpha_l)
\right)\right].
\end{equation}
Here  $\lnk$ denotes the Gau{\ss} linking number. One nontrivial aspect
of the above result is the appearance of the \emph{self-linking}
$\lnk(\alpha,\alpha)$ of loops which is defined as Gau{\ss} linking of $\alpha$ 
with a slightly displaced loop $\alpha'$. The displaced loop is defined using a
\emph{framing} of $\alpha$. This is the expression, in the U(1) case, of the well known fact that a choice of framing is necessary to
compute expectation values of holonomies in Chern-Simons theory. It seems to
present an obstacle to using the mathematical methods cited above, in which only unframed loops are considered. One way around it may be to use and extend the formalism of framed spin-networks \cite{Major:1995yz} which provides for framed loops.

Here we will take another route. We will give a derivation of the expectation values
\eqref{eq_first}, using techniques related to generalized connections. 
The idea is to use Stokes' theorem to replace ho\-lo\-no\-mies under the path integral by (exponentials of) integrals of the curvature over surfaces that have the holonomy loops as boundaries. Then we use the property of Chern-Simons theory
that the functional derivative of the action with respect to the connection
yields the curvature of the connection. We can thus replace the
above-mentioned curvature integrals under the path integral by functional
derivatives. Finally we observe that these functional derivatives are well
defined objects, known in the loop quantum gravity literature as \emph{flux operators} 
\cite{Ashtekar:1997eg} and can be evaluated, to yield the desired expectation
values. On the technical level, it is a connection between Gau{\ss} linking and intersection numbers observed quite some time ago
\cite{Ashtekar:1997rg} that leads to the appearance of the linking numbers in the result.  
As one can see a prominent role in this derivation is played by surfaces
that have a given loop as a boundary. Such surfaces exist for any given loop,
but they are not unique. We will demonstrate that the choice of such a surface
precisely amounts to the choice of a framing. What is more, we find that there are rules for assignments of surfaces to loops such that the induced framing \emph{only} depends on the topological properties of the loops. Formulated differently, there are ways to assign a framing to a loop which depends just on its topological properties. Such  \emph{natural} framings can be incorporated into the definition of the path integral, which, in turn becomes well defined on unframed loops. In this way, we define the
Chern-Simons path integral as a path integral\footnote{On the technical
level, what we obtain is a well defined functional on gauge invariant
cylindrical functions.} on generalized connections.

The structure of the article is as follows: In the next
section we will introduce some background material and explain our strategy. 
In section \ref{se_abelian} we carry out the calculation of the expectation
values. We describe two examples of natural framings in section \ref{se_framing}, and close with a short discussion of 
our results in section \ref{se_closing}. An appendix collects some definitions
and technical results related to the surfaces that we use.
\section{Strategy}
In the present section we will describe our strategy to
define the CS path integral. Let us fix the manifold to be
$M=\R^3$ and start by listing our ingredients.
\begin{itemize}
\item[(a)] We denote by $\abar$ the space of generalized U(1) connections
(see for example \cite{Ashtekar:1994mh}).
This is a space of distributional connections. It is compact, Hausdorff,
smooth connections are dense, and it is well suited for measure theory.

\item[(b)] We denote by 
$\mal$ the Ashtekar-Lewandowski measure \cite{Ashtekar:1994mh}, a non-degenerate, uniform measure
on the space $\abar$. 

\item[(c)] A rigorous definition of the functional
derivative $\delta/\delta A$ has been given \cite{Ashtekar:1997eg}. More precisely,
\begin{equation*}
\wh{X}_{S}=\int_S \epsilon_{cab}\frac{\delta}{\delta A_c}\text{d}x^a\wedge
\text{d} x^b
\end{equation*}
for any surface $S$ gives a well defined derivation on suitably differentiable
functionals of $\abar$. We will describe its action in more detail below. Here
we only need to mention that it has the expected
adjointness properties with respect to the scalar product induced by $\mal$:
\begin{equation*}
\wh{X}_{S}^\dagger = - \wh{X}_{S}.
\end{equation*}

\item[(d)] Stokes' theorem relates the contour integral
of an abelian  connection $A$ around a loop $\alpha$
bounding a surface $S$ (as in figure \ref{fi_unknot_seifert}) to the integral of the curvature
over the surface $S$:\footnote{Since we work on $\R^3$ we can assume that the
bundle is trivial and work with the connection as a one-form on the base
manifold.}
\begin{equation*}
\oint_{\alpha} A = \iint_S \text{d}A.
\end{equation*}
It holds for a very general class of surfaces (so called \emph{domains of integration}, see for example \cite{CWD}), and in particular for the types of surfaces 
considered in the present work. 
\begin{figure}
\centerline{\epsfig{file=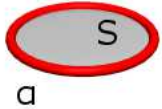, scale=0.5}}
\caption{\label{fi_unknot_seifert} A simple Wilson loop
$\alpha$ and a surface $S$ bounded by it}
\end{figure}
\end{itemize}
Our strategy is now as follows. To evaluate the expectation
value of a Wilson line, we rewrite it as a surface integral over the curvature
associated to the connection. Then we observe that
\begin{equation*}
\frac{\delta S_\text{CS}}{\delta A_c} = \frac{k}{4\pi}F_{ab}\epsilon^{abc}
\end{equation*}
which we use to replace the curvature under the CS path
integral with functional derivatives. Finally we apply
integration by parts to these derivatives. These
manipulations are well motivated but formal. We will see
however that they lead to an expression that is well
defined and can be evaluated in order to obtain the desired
expectation value. Let us be more explicit and consider the
CS expectation value of the product of a holonomy 
\begin{equation*}
h_\alpha[A]\doteq\exp\left[\oint_\alpha A\right]
\end{equation*}
with another functional $F[A]$ of the connection. We stress that the 
loop $\alpha$ can be arbitrarily knotted.\footnote{In fact, throughout the text we could also use the word \emph{knot}
in place of \emph{loop}, since all the loops are piecewise analytic and hence equivalent to polygonal loops.} 
The formal manipulations just described work out as
follows:
\begin{align*}
\expec{h_\alpha}&=\int_{\abar} F[A] h_\alpha[A] \exp(i
S_{\text{CS}}[A])\, \dmal [A]\\
&=\int_{\abar}F[A]\exp\left[\iint_S F\right]\exp(i S_{\text{CS}}[A])\, \dmal [A]\\
&=\int_{\abar}F[A]\exp\left[
\frac{4\pi}{k}\iint_{S} 
\overrightarrow{\frac{\delta}{\delta A}}
\right]
\exp(i S_{\text{CS}}[A])\, \dmal [A]\\
&=\int_{\abar}F[A]\exp\left[
-\frac{4\pi}{k}\iint_{S}
\overleftarrow{\frac{\delta}{\delta A}}\right]
\exp(i S_{\text{CS}}[A])\, \dmal [A]\\
&=\expec{\exp\left(\frac{4\pi}{k}\wh{X}_{S}\right)F[A]}
\end{align*}
where $S$ is a surface bounded by $\alpha$ (as in figure
\ref{fi_unknot_seifert}). In the last line we have
expressed our result as the expectation value of a
functional differential operator acting on the functional $F$. 
As we have indicated above, this operator is a well defined derivation on 
cylindrical functions. Therefore, if $F[A]$ is, for example, a product of 
holonomies, it can be easily evaluated. In this way, we will be able 
to recursively calculate the desired expectation values. 

Let us make an important remark regarding the strategy sketched above. It addresses a question that
the reader may have had while reading our description of the surface $S$ bounded by the loop
$\alpha$: Could it not happen that $\alpha$ does not bound \textit{any} surface? What then? Indeed, by
definition, if $\alpha$ is a non-trivial cycle of the manifold $M$, it is not the boundary of a
surface. This is why we restrict ourselves to $\R^3$ for the present paper.
What our approach can say in the case of not simply connected manifolds is an
interesting question that may be investigated elsewhere.

Still, even for the case of $M=\R^3$ it may not be immediately obvious that
for any loop, one can find a surface bounded by it. This however is
assured: there are always such surfaces. One type of surface having a given loop (or even link) as its boundary is called 
\emph{Seifert surface} of the knot or link. By definition a Seifert surface is embedded, connected, and orientable, the latter of which is important for our purposes since it ensures that the derivation $X$ is well defined. A theorem by Pontrjagin and Frankl asserts that there exists a Seifert surface for any knot or link. (An elegant construction of such a surface is due to Seifert, hence the name.) Two examples of Seifert surfaces are depicted in figure \ref{fi_seifert_example}.
Besides Seifert surfaces, we will also use the notion of a \emph{compressing disc} for a given loop. A compressing disc for a loop is a surface bounded by the loop, which is an immersion of a disc. More detailed descriptions of the surfaces used in the following, as well as some simple technical results, are compiled in the appendix.

Now that we have ensured that to a given knot or link there are surfaces bounding
it, a reasonable question is: Aren't there too many? We will see that the non-uniqueness in assigning a surface to a loop shows
up in the end result as a choice of framing for the loop. A surface bounded by the loop induces a framing on the loop\footnote{Pick any vector field on the loop that is everywhere non-zero, and transversal to the surface. Any such choice will lead to a framing, and all the framings obtained this way are equivalent.}, and it is this framing that determines the self-linking in \eqref{eq_first}.    
\begin{figure}
\centerline{\epsfig{file=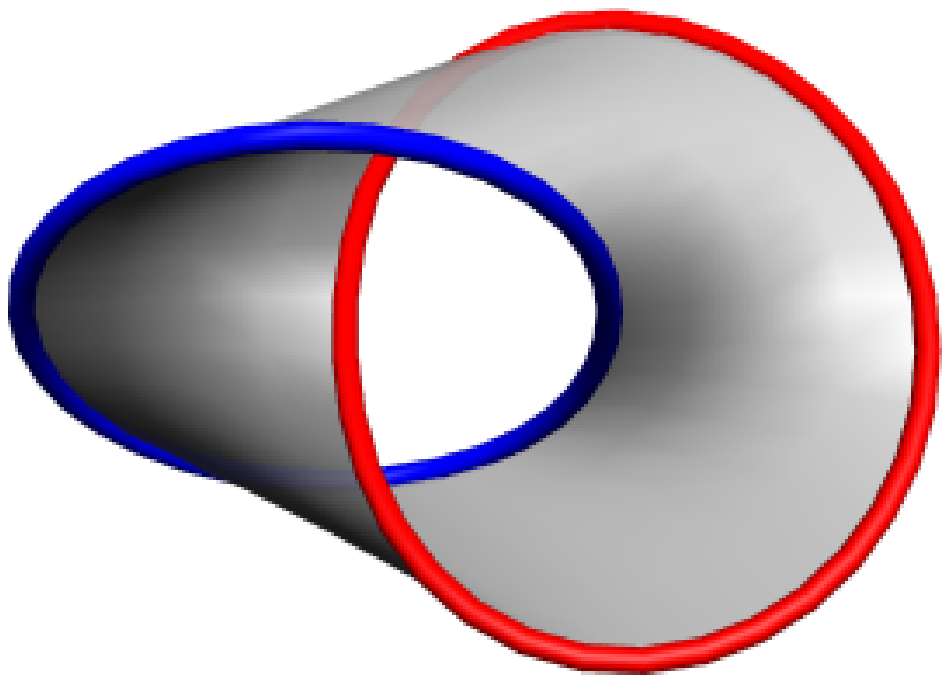, scale=0.4}
\epsfig{file=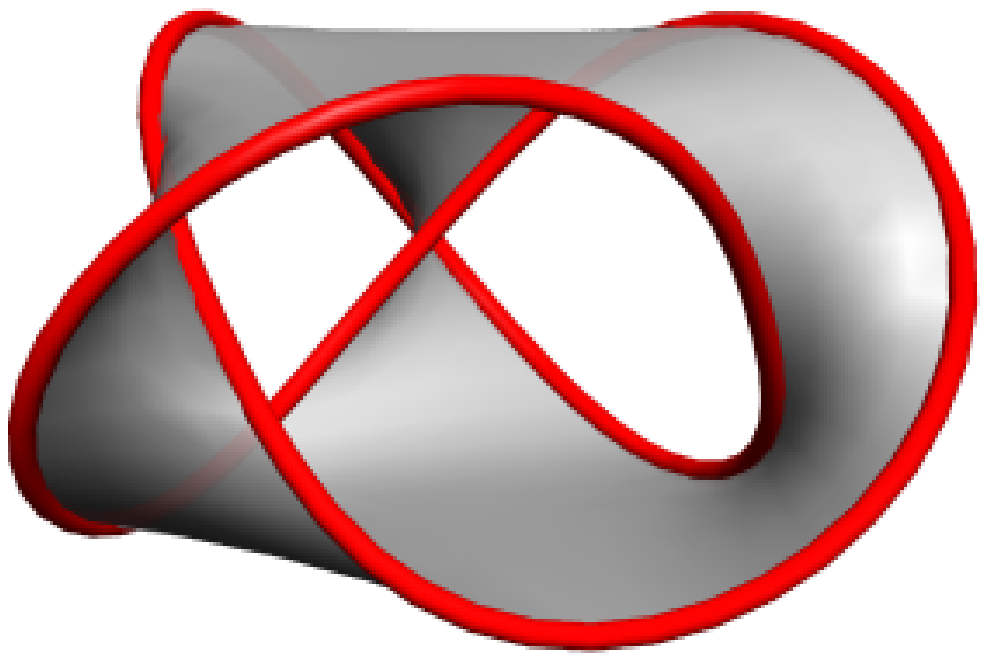, scale=0.4}}
\caption{\label{fi_seifert_example} Examples for Seifert
surfaces of knots and links: For the trefoil knot (right),
and the Hopf rings (left) (graphics created with \texttt{SeifertView} \cite{SV})}
\end{figure}

With these remarks in place, we can now move towards the actual calculation.  

\section{Calculation of expectation values}
\label{se_abelian}
In this section we will obtain the expectation value \eqref{eq_first} using language and some techniques from 
loop quantum gravity. For some technical background on the surfaces used we refer to the appendix. 
We will work with the action
\begin{equation*}
S_\text{CS}[A]=\frac{k}{4\pi} \int_{\R^3}A\wedge\text{d}A= \frac{k}{8\pi} \int_{\R^3} \epsilon^{abc}A_a\partial_bA_c\,\text{d}^3x
\end{equation*}
and $A$ is a real field. Then
\begin{equation*}
\frac{\delta S_\text{CS}}{\delta A_c} = \frac{k}{4\pi}F_{ab}\epsilon^{abc}
\end{equation*}
with $F_{ab}$ the components of curvature $F=dA=\partial_aA_b\text{d}x^a\wedge \text{d} x^b$, hence
\begin{equation}
\label{eq_der}
\epsilon_{cab}\frac{\delta S_\text{CS}}{\delta A_c}\text{d}x^a\wedge \text{d} x^b=\frac{k}{2\pi}F.
\end{equation}
The integrated functional derivative on the left hand side is a well known object in
loop quantum gravity. It can be applied to holonomies, or, more general,
cylindrical functions and acts as a derivation $\wh{X}_S$:
\begin{equation*}
\wh{X}_{S}=\int_S \epsilon_{cab}\frac{\delta}{\delta A_c}\text{d}x^a\wedge \text{d} x^b.
\end{equation*}
For the group $U(1)$, an integer $n$ and a loop $\alpha$,
\begin{equation}
\label{eq_comm}
\wh{X}_{S}h^n_{\alpha}=i n I(S,\alpha) h^n_{\alpha}
\end{equation}
where $I(S,\alpha)$ is the signed intersection number between $S$ and $\alpha$ (see appendix \ref{ap_surfaces}).
We also note that
\begin{equation*}
I(S,\alpha)=\lnk(\partial S, \alpha).
\end{equation*}
For a proof see the appendix. We rewrite the right hand side of \eqref{eq_der} using Stokes' theorem, and obtain
(formally):
\begin{equation*}
\frac{2\pi}{k}\wh{X}_{S} e^{iS_\text{CS}}= ie^{iS_\text{CS}}\int_{\partial S} A\equiv i A_{\partial S}e^{iS_\text{CS}}
\end{equation*}
Thus it is useful to introduce the operator $A_\alpha$ with commutation relations
\begin{equation}
\label{eq_comm2}
[\wh{X}_S, A_\alpha]= I(S, \alpha).
\end{equation}
Strictly speaking $A_\alpha=-i\ln h_{\alpha}$ is not well defined since 0 is in the discrete spectrum of $h_\alpha$. 
But we will not be concerned by this here, since some of the manipulations under the path integral are formal anyway. 
The commutator \eqref{eq_comm2} is zero when $\alpha=\partial S$, since $I(S, \alpha)$ as we have defined it is zero in this case.
This is the regularization chosen in loop quantum gravity, but one can make other choices, and we will
do so, here. A surface $S$ endows its boundary loop $\partial S$ with a framing in a natural way. Just choose
a smooth vector field on $\partial S$ that is transversal to $\partial S$ and nowhere tangent to the surface.
Different vector fields chosen this way are equivalent as framings, as one can easily see.
Pick one of these vector-fields, $v$, and use it to ``transport'' the boundary loop outwards: If $(\partial S)(t)$
is some parametrization of the boundary loop, then $(\partial S)_\epsilon(t):= (\partial S)(t)+\epsilon v(t)$.
Let us then define
\begin{equation}
\label{eq_reg}
[\wh{X}_S, A_{\partial S}]= \lim_{t\rightarrow 0}[\wh{X}_S, A_{\partial S_{\epsilon}}].
\end{equation}
Let us also modify our definition of the signed intersection number $I$ such that with the definition \eqref{eq_reg},
the commutation relations between the $\wh{X}$ and the connection can still be written in the
form \eqref{eq_comm}. We also note that \eqref{eq_reg} leads back to the standard regularization (in which the
commutator \eqref{eq_reg} vanishes) for surfaces without self-intersections.

Now we can calculate the the expectation values we are interested in. Let $\alpha_1\ldots \alpha_N$ be loops,
$S_1\ldots S_N$ surfaces with $\partial S_i=\alpha_i$, and  $n_1\ldots n_N$ integers.
Then formally
\begin{align*}
\expec{\prod_i h_{\alpha_i}^{n_i}}&=\int\dmal[A] h_{\alpha_N}^{n_N}\ldots h_{\alpha_2}^{n_2}
\left[\sum_{l=0}^\infty \frac{1}{l!}(in_1A_{\alpha_1})^l\right]e^{i S_{\text{CS}}}\\
&=\int\dmal[A] h_{\alpha_N}^{n_N}\ldots h_{\alpha_2}^{n_2}
\left[1+\sum_{l=1}^\infty \frac{1}{l!}(in_1A_{\alpha_1})^{l-1} \frac{2\pi n_1}{k}\wh{X}_{S_1}\right]e^{i S_{\text{CS}}}\\
&=\int\dmal[A] h_{\alpha_N}^{n_N}\ldots h_{\alpha_2}^{n_2}
\left[1-\frac{2\pi n_1}{k}\sum_{l=1}^\infty \frac{1}{l!}(in_1A_{\alpha_1})^{l-1} \overleftarrow{\wh{X}_{S_1}})\right]e^{i S_{\text{CS}}}\\
\end{align*}
where in the last step we have used that $\wh{X}_{S_1}$ is anti-symmetric with respect to
the measure $\mal$. As a consequence the derivation now acts on something that is allowed to act, and the next steps commute
it through the holonomy operators. From \eqref{eq_comm2},\eqref{eq_comm}
\begin{equation}
\label{eq_rec}
-\frac{2\pi n_1}{k}(in_1A_{\alpha_1})^{l-1}\overleftarrow{\wh{X}_{S_1}}=-\frac{2\pi n_1}{k}
\left[(l-1)in I(S_1,\partial S_1)(inA_{\alpha_1})^{l-2}+ \overleftarrow{\wh{X}_{S_1}}(inA_{\alpha_1})^{l-1}\right].
\end{equation}
with the result that we have moved the derivation past all the connection terms coming from the same loop $\alpha_1$.
Before moving the derivation further left, we want to repeat the procedure of exchanging an $A_{\alpha_1}$ for a
$\wh{X}_{S_1}$ and commuting it left through all the $A_{\alpha_1}$'s, until all of the $A_{\alpha_1}$'s have
been eliminated this way. To this end we will use the recursion formula implicit in \eqref{eq_rec}: It says that
under the manipulations described
\begin{equation}
\label{eq_rec2}
(in_1A_{\alpha_1})^l=:F(l)=(k-1)X F(l-2)+E F(l-1)
\end{equation}
where we have introduced the shortcuts
\begin{equation*}
X=-\frac{2\pi n_1^2i}{k}I(S_1,\partial S_1)\qquad E=-\frac{2\pi}{k}\overleftarrow{\wh{X}_{S_1}}.
\end{equation*}
Let us also make the definition $F(0)=1$.
Note that with the understanding that the operators $\wh{X}_{S_1}$ always stand to the left, we can
use \eqref{eq_rec2} as if $E$ is a number. \eqref{eq_rec2} is not easily solved explicitly, so we will
work with its exponential generating function
\begin{equation*}
G(w):=\sum_{l=0}^\infty\frac{1}{l!}F(l)w^l.
\end{equation*}
We note that what we are really interested in is
\begin{equation*}
G(1)=\sum_{l=0}^\infty\frac{1}{l!}F(l).
\end{equation*}
From \eqref{eq_rec2} one can derive that $G(w)$ must satisfy a differential equation:
\begin{equation*}
G'(w)=XwG(w)+EG(w)
\end{equation*}
The initial condition is $G(0)=F(0)=1$. We find the solution
\begin{equation*}
G(w)=\exp\left[wE+\frac{w^2X}{2}\right].
\end{equation*}
By evaluating at $w=1$ we find that under the path integral, and using
the manipulations we have described above,
\begin{equation*}
\sum_{l=0}^\infty\frac{1}{l!}(in_1A_{\alpha_1})^l=e^{-\frac{2\pi}{k}\overleftarrow{\wh{X}_{S_1}}}e^{-\frac{\pi n_1^2i}{k}I(S_1,\partial S_1)}.
\end{equation*}
Inserting this into the calculation of the expectation value, and pulling the $\wh{X}_{S_1}$ further to the left, we end up with
\begin{equation*}
\expec{\prod_i h_{\alpha_i}^{n_i}}=\exp\left[-\frac{\pi n_1^2i}{k}I(S_1,\partial S_1) - \frac{2\pi}{k}\sum_{j=2}^{n}I(S_1,\alpha_j)\right]
\int\dmal[A] h_{\alpha_N}^{n_N}\ldots h_{\alpha_2}^{n_2}e^{i S_{\text{CS}}}.
\end{equation*}
We can now repeat this procedure with the other holonomies, and, as a final step replace signed intersection numbers $I$ by Gau{\ss}
linking (or, in the case of $I(S, \partial S)$, self-linking ) and obtain the well known result
\begin{equation}
\label{eq_abelianresult}
\expec{\prod_i h_{\alpha_i}^{n_i}}=\exp\left[-\frac{\pi i}{k}\left(\sum_j n_{j}^2\lnk(\alpha_j,\alpha_j)+2\sum_{j<l}n_jn_l\lnk(\alpha_j,\alpha_l) \right)\right].
\end{equation}
Let us summarize the main points of this section:
\begin{itemize}
\item CS expectation values can be calculated using some formal manipulations that are however well motivated from the properties
of the objects involved.
\item Framing and self-linking enter the picture via choice of surfaces that have a given
loop as boundary.
\end{itemize} 
Some remarks: (i) The linking as defined above does not take into account the framing when
considering two loops that partially, but not completely, overlap. It is an interesting question
if this can be changed while maintaining consistency.\\
(ii) As far as we can see the formal calculation performed above and its result are insensitive to the order in which
the connection terms are exchanged for derivations. They are however not insensitive to more drastic changes. For example
one could consider using a Seifert surface for the entire link ${\alpha_1\cup\alpha_2\ldots\cup\alpha_n}$, instead
of separate surfaces for each loop. In that case, we would find the expectation value to be 1 identically. So this part
of the procedure has to be regarded as a regularization ambiguity that we have fixed in a certain way.
\section{Natural framings}
\label{se_framing}

Obviously, the result \eqref{eq_abelianresult} depends on the surfaces $S_i$ chosen in the process of calculation. These enter through the self-linking
for the loops $\alpha_i$ that they define. The formalism of loop quantum gravity is using loops (or more generally, graphs) that are
not framed. So for \eqref{eq_abelianresult} to serve as a definition of a functional on $\cyl$, one will have to make a choice
of surface, and hence of framing, for each loop. One way to proceed is to leave the framing unspecified and go over to the formalism of framed spin-networks \cite{Major:1995yz}. But there is also another option: One can search for ways in which loops get assigned surfaces (and hence framings) based on their properties. Such a choice would then be part of the definition of the path integral. A minimal requirement on such a choice is that it makes the expectation values \eqref{eq_abelianresult} invariant under diffeomorphisms, i.e., 
\begin{equation*}
\expec{\prod_i h_{\alpha_i}^{n_i}}=\expec{\prod_i h_{\phi(\alpha_i)}^{n_i}}
\end{equation*}
for any diffeomorphism $\phi$ of $\R^3$, which can be connected to the identity. A necessary and sufficient condition in terms of a map  $\alpha \longmapsto S_\alpha$ from loops to surfaces
is thus 
\begin{equation*}
I(S_\alpha,\alpha)=I(S_{\phi(\alpha)},\phi(\alpha))
\end{equation*}
for all loops $\alpha$ and all diffeomorphisms $\phi$ connected to the identity. Such a map then endows each loop with a framing such that  
\begin{equation}
\lnk(\alpha,\alpha)=\lnk(\phi(\alpha),\phi(\alpha))
\label{eq_nat}
\end{equation}
for all loops $\alpha$. Let us mention two examples for such assignments of surfaces and hence framings:
\begin{enumerate}
\item For a loop $\alpha$ choose $S_\alpha$ to be a minimal compressing disc, i.e.\ one that minimizes the number of intersections of the loop
with the surface. The minimal number of such intersections is an invariant of the loop, the knottedness $\knot(\alpha)$ \cite{Greene}, and
$\lnk(\alpha,\alpha)=\knot(\alpha)$.
\item For a loop $\alpha$ choose $S_\alpha$ to be a Seifert surface. A Seifert surface does not self-intersect and hence $\lnk(\alpha,\alpha)=0$.
\end{enumerate}
We call the framing obtained from these and similar prescriptions \emph{universal} and \emph{natural}, the former because they encompass all loops, the latter because of their covariance \eqref{eq_nat}. (The second prescription also yields what is called \emph{natural framing} in the mathematical literature.) 
Making one of these choices, \eqref{eq_abelianresult} becomes a function of unframed loops. Still it is not a functional on
$\cyl$ since not all functionals in $\cyl$ are linear combinations of multiloops. It is however well known that the gauge invariant functionals can all be written as such linear combinations and on those \eqref{eq_abelianresult} defines a well defined functional. 

We remark that for graphs that can be decomposed into loops, the above prescriptions will also give a notion of `framed graph', and hence also a notion of framing for gauge invariant functionals in $\cyl$. It would be interesting to compare this in detail to the notion of framed spin network of \cite{Major:1995yz}. We suspect that the notions will turn out to be the same. 
\section{Closing remarks}
\label{se_closing}
In the present article we have done two things. On the one hand we have explained how the Wilson loop expectation value \eqref{eq_first} of U(1) Chern-Simons theory can be interpreted as a functional over gauge invariant cylindrical functions by using a universal framing prescription, and we have given two examples of such a prescription. On the other hand we have \emph{derived} the expectation values \eqref{eq_first} using some heuristic manipulations of the path integral of the theory. These manipulations used the fact that certain functional derivatives have well defined action on Wilson loops in the form of \emph{flux operators}, well known from loop quantum gravity.  

Since U(1) Chern-Simons theory is not very interesting in itself, the potential value of the work presented above lies in the methods used. Using a universal framing prescription to describe a theory that needs framing with the mathematical methods based on generalized connections may be useful for different theories. And we have used a regularization of the (abelian) flux that differs from the one used in LQG. This may also find applications elsewhere. But, most importantly, the technique we have used to compute the expectation values \eqref{eq_first} 
seem to be applicable also to Chern-Simons theory with compact, \emph{non}-abelian gauge group. In particular, in \cite{Sahlmann:2011uh} we have started to investigate the case of SU(2), relevant to knot theory \cite{Witten:1988hf}, Euclidean gravity in three dimensions, and maybe even to four dimensional gravity via the Kodama state. 
Our findings show that again flux operators can be defined that replace holonomies under the path integral, using a non-abelian generalization of Stokes' theorem. The regularization of these operators is however much more complicated than in the abelian case. In particular, since the functional derivatives do not commute for the non-abelian group, there is an ordering ambiguity. In fact, something like it is expected since it may be the technical reason for the occurrence of quantum SU(2) in the quantum theory. In \cite{Sahlmann:2011uh}, we show that a certain natural mathematical structure in the theory of Lie, algebras -- the so-called \emph{Duflo map} -- can be used to perform the ordering, and we calculate the results for some simple cases. 
\section*{Acknowledgements}

HS gratefully acknowledges funding for the earlier stages of this work through a
Marie Curie Fellowship of the European Union. His later research was partially supported by the Spanish MICINN project No.\ FIS2008-06078-C03-03.
\begin{appendix}

\section{Surfaces, Gau{\ss} linking, and intersection numbers}
\label{ap_surfaces}
Before we start the calculation of expectation values in CS theory, we will
collect here some definitions and mathematical facts needed in the following sections.

For the moment a loop will be a smooth, compact closed one dimensional submanifold in
$\R^3$. Also we assume loops to be tame, i.e. equivalent to a polygonal knot. Later
we will have to make some adjustments to this definition owing to the fact that
loops in the formalism of loop quantum gravity that we are going to use are actually
piecewise analytic.

Let us pick an orientation of $\R^3$ and
stick with it throughout. All the surfaces we consider will be orientable,
and we assume them to be oriented, even if we do not state this explicitly
each time. We will also take the loops to be oriented.

Given a loop $\alpha$ in $\R^3$ we consider surfaces
$S$ such that $\partial S=\alpha$. We will always assume that
$S$ and $\alpha$ are oriented consistently.\footnote{That means, if $T$ is a positively oriented tangent vector to the loop
$\alpha$ at a point $p$, and $N$ is the outward normal vector (tangent to $S$) of $S$ at $p$, then $(N,T)$ is a positively oriented
basis of $T_pS$.} We will use two classes of
such surfaces, Seifert surfaces, and compressing discs. Let us give a
brief description of each.
\begin{defi}A Seifert surface for a loop $\alpha$ is an orientable, connected submanifold $S$
such that $\partial S=\alpha$.
\end{defi}
Seifert surfaces exist for any loop in $\R^3$. In fact there exist a simple
algorithm to construct one for a given loop. There is however no uniqueness, not even in a
topological sense: A loop has many different Seifert surfaces. For these and other general results on Seifert surfaces see for example \cite{M}.

For the definition of a compressing disc we follow \cite{Greene}. That reference
works with loops in $S^3$ instead of $\R^3$ but for the few things we need
from there, the compactification makes no difference.
\begin{defi}
We define a compressing disc of a loop $\alpha$ to be a map $f$ from the
two dimensional disc $D^2$ into $\R^3$ such that $f\rvert_{\partial D^2}=\alpha$
and $f\rvert_{\text{int} D^2}$ is transverse to $\alpha$.
\end{defi}
Then $f({\text{int} D^2})$
has only finitely many intersections with $\alpha$, and one can show that
for a given loop there is a minimal number of intersections that can be achieved
by varying the compressing disc. This number is called knottedness and
is an invariant of the loop. It is shown in \cite{Greene} that starting from a
compressing disc with the image of which has $n$ intersections with the loop,
one can always find a compressing disc with $n$  or less intersections, which in
addition is an immersion of $D^2$ into $\R^3$. In the following (and in the main text)
we will always assume all the compressing discs to be immersions.

Let us call an orientable, connected, two-dimensional submanifold of $\R^3$ a surface of \emph{type I}, and
an immersion of $D^2$ into $\R^3$ a surface of \emph{type II}. A notion we need
for both types of surfaces is that of the \emph{signed intersection number} between a
surface and a loop.
\begin{defi}
For a surface $S$ of type I and a loop $\alpha$, the signed intersection number is
\begin{equation*}
I(S,\alpha)=\sum_{p\in S\cap\alpha} \kappa(p)
\end{equation*}
where $\kappa(p)=+1$ if the intersection is transversal and the orientations of $\alpha$ and $S$ together
give the orientation of the one chosen on $\R^3$,\footnote{We mean: If $T$ is a positively oriented tangent vector to $\alpha$ at the intersection point $p$,
and $(v_1,v_2)$ is a positively oriented basis of $T_pS$, then $(T,v_1,v_2)$ is positively oriented in $T_p\R^3$.} $\kappa(p)=-1$ if the intersection is transversal and the
orientations of $\alpha$ and $S$ together
give the orientation opposite of the one $\R^3$, and zero otherwise.
\end{defi}
For a surface of type II the definition is almost the same, except that intersections
with a loop at a point of self-intersection of the surface may count multiple times.
\begin{defi}
Let $f:D^2\rightarrow \R^3$ be a surface of type II. Then we can find an open cover $\{U_J\}$ of $D^2$
such that $f(U_J)$  is a surface of type I for every $J$. Then given a loop $\alpha$ we define
\begin{equation*}
I(S,\alpha)=\sum_J I(f(U_J),\alpha).
\end{equation*}
\end{defi}
It is easily checked that this definition is independent of the open cover.

\begin{figure}
\centerline{\epsfig{file=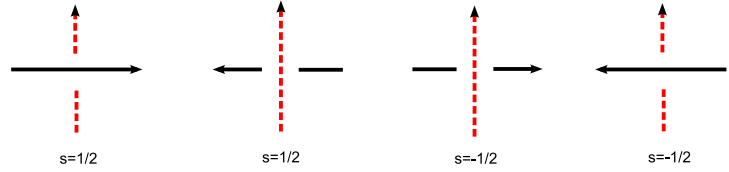, scale=0.6}}
\caption{\label{fi_linking} The types of crossings used in the calculation of the linking number of two loops.}
\end{figure}
Another notion that we will need is that of the Gau{\ss}-linking number (or simply linking number). It is a property of a pair of oriented loops. One of several
equivalent ways of defining it is the following: Given two oriented loops $\alpha_1,\alpha_2$, draw a diagram of the pair of loops. To each crossing $c$ in the diagram,
determine the quantity $s_c$ by comparison with figure \ref{fi_linking}. Then the linking number is
\begin{equation*}
\lnk(\alpha_1,\alpha_2)=\sum_{c} s_c.
\end{equation*}
It is easy to see that the linking number is independent of the diagram chosen, and that it is, in fact, a topological invariant of the loops.

\begin{figure}
\centerline{\epsfig{file=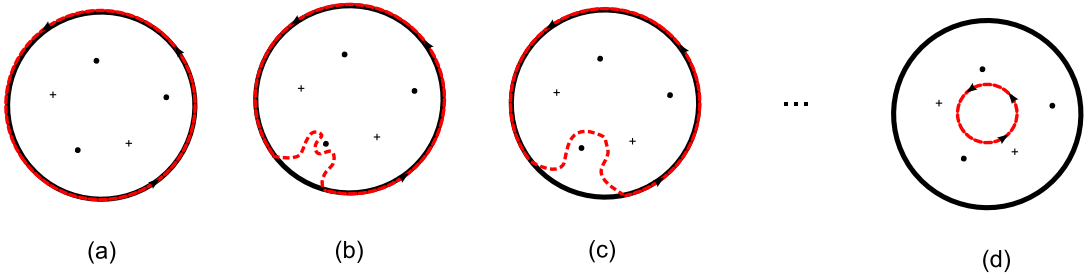, scale=0.8}}
\caption{\label{fi_proof} The shrinking procedure of the loop used in the proof of prop.\ \ref{pr_link-intersection}. Orientation of $D^2$ is out of this page,
towards the reader, intersections with positive signature are indicated by a `+', those with negative signature by a dot.}
\end{figure}
The linking number is related to the signed intersection number as follows:
\begin{prop}
\label{pr_link-intersection}
Let $S$ be a surface of type I or II, and $\alpha$ be a loop. Then
\begin{equation*}
I(S,\alpha)=\lnk(\partial S,\alpha).
\end{equation*}
\end{prop}
\begin{proof}
For a surface of type I, the proof of the proposition is contained in the paper \cite{Ashtekar:1997rg} as appendix 1, and so we won't reproduce it here.
For surfaces of type II the argument is as follows: Let a type II surface $S$ and a loop $\alpha$ be given. Thus we have an immersion $f$ of $D^2$
into $\R^3$, the image being $S$. Let us describe the situation by depicting $D^2$, together with the preimage of the intersection points of $S$ with $\alpha$.
Thus we get a picture like figure \ref{fi_proof}(a), where we have additionally kept track of the signature $\kappa$ of the intersection. The signed intersection
number can be read of from the diagram, by subtracting the number of dots (preimages of negative signature intersections) from the number of pluses (preimages of
positive signature intersections). Now we will construct a family of new type II surfaces $S(t)$ from $S$, by suitably restricting the domain of the immersion $f$, or, in
other words by shrinking the boundary loop. To this end, we consider a homotopy of loops $\beta(t)$ in $D^{2}$, with $\beta(0)=\partial D^2$ and $\beta(1)$ a loop
that has no preimages of intersections within the disc that it bounds (such a loop is depicted in \ref{fi_proof}(d)), and such that its image under $f$
has no linking with $\alpha$. Obviously the latter two conditions are compatible, and many such loops exist.
We define the surface $S(t)$ as the
image under $f$ of the disc bounded by $\beta(t)$. Since $f$ is only an immersion, $f(\beta(t))$ for any given $t\neq 0$ is not necessarily
a loop: It can have self-intersections. One can convince oneself, however, that one can choose the homotopy in such a way that the
self-intersections of $f(\beta(t))$, if present at all, are isolated points, and such that $f(\beta(1))$ does not have any self-intersections.
We will consider such a homotopy in the following. Let us keep track of $\lnk(\alpha,\beta(t))$ and $I(S(t),\alpha)$ as we vary $t$
Obviously we start with
\begin{equation*}
\lnk(\alpha, f(\beta(0)))=\lnk(\alpha,\partial S)\qquad I(S(0),\alpha)=I(S,\alpha).
\end{equation*}
As long as $f(\beta(t))$ does not develop self-intersections or $\beta(t)$ moves past the pre-image of an intersection of $S$ with $\alpha$, nothing changes.
When $f(\beta(t))$ develops a self-in\-ter\-sec\-tion, $\lnk(\alpha,\beta(t))$ is a priori not well defined. But we will define it to be whatever one gets when one
removes the self-intersections by slightly moving $f(\beta(t))$ such that the self-intersections disappear, without however moving
$f(\beta(t))$ through $\alpha$. The result is independent of the precise way this done, since $\lnk$ can be computed using only the crossings
of $\alpha$ with $f(\beta(t))$, no self-crossings of $f(\beta(t))$. $I(S(t),\alpha)$ is obviously insensitive to self-intersections of $f(\beta(t))$.
If on the other hand $\beta(t)$ moves past the pre-image of an intersection of $S$ with $\alpha$ (as depicted in figures  \ref{fi_proof}(b),(c)), both
quantities change: For a positive signature intersection,
\begin{equation*}
I(S(t_\text{after}),\alpha)= I(S(t_\text{before}),\alpha) - 1
\end{equation*}
and, by studying figure \ref{fi_linking},
\begin{equation*}
\lnk(f(\beta(t_\text{after})),\alpha)= \lnk(f(\beta(t_\text{before})),\alpha) - 1.
\end{equation*}
Similarly, for passing a negative signature intersection,
\begin{align*}
I(S(t_\text{after}),\alpha)&= I(S(t_\text{before}),\alpha) +1\\
\lnk(f(\beta(t_\text{after})),\alpha)&= \lnk(f(\beta(t_\text{before})),\alpha) +1.
\end{align*}
At the end of the process,
\begin{equation*}
\lnk(\alpha,f(\beta(1)))=0,\qquad I(S(1),\alpha)=0.
\end{equation*}
Thus $\lnk(\alpha,\beta(t))$ and $I(S(t),\alpha)$ are the same in the end, and they change in step, thus they have been equal at the beginning as well, which
proves the proposition.
\end{proof}
A final remark of this appendix concerns the differentiability category used for the loops and surfaces. In the description above all surfaces (and implicitly also their boundaries) are smooth. In loop quantum gravity the surfaces for which the flux operators are well defined are however real analytic (or piecewise real analytic, defined in a suitable sense), and so are the loops. This is to ensure that 
there are only finitely many transversal intersections for all pairs of compact loops and surfaces, which is in turn necessary to make the analog, for non-abelian groups, of the commutator \eqref{eq_comm} well defined.  This is however not a concern for the work presented here, 
since due to the abelian nature of U(1), flux operators are well defined even on loops that intersect the underlying surface infinitely often 
as long as the Gauss linking between the loop and the boundary of the surface is finite. This is the case if we work with piecewise analytic loops and boundaries. Thus the remaining question is whether the existence of Seifert surfaces and compressing discs, as well as the results 
of this appendix continue to hold for piecewise analytic loops. It is easy to see that Seifert surfaces continue to exist in this case, and 
the notion of a compressing disc can be trivially generalized to allow for a piecewise analytic boundary. Finally it is easy to see that all properties remain intact. Thus, in the main text, we will always assume loops to be piecewise analytic and surfaces to be smooth.  
\end{appendix}


\end{document}